\documentclass[12pt]{amsart}

\usepackage{amsfonts,amsmath,amsthm}
\usepackage{latexsym}

\newtheorem{theorem}{Theorem}[section]
\newtheorem{corollary}[theorem]{Corollary}
\newtheorem{lemma}[theorem]{Lemma}
\newtheorem{remark}[theorem]{Remark}
\newtheorem{problem}[theorem]{Problem}

\newtheorem{definition}{Definition}[section]

\theoremstyle{definition}
\theoremstyle{problem}
\theoremstyle{conjecture}

\newcommand{\dst}{\displaystyle}

\newcommand{\TT}{\ensuremath{\mathbb{T}}}
\newcommand{\ZZ}{\ensuremath{\mathbb{Z}}}
\newcommand{\R}{\ensuremath{\mathbb{R}}}

\newcommand{\Co}{\ensuremath{\mathbb{C}}}

\def \C {\mathbb{C}}

\newcommand{\ac}{\ensuremath{\mathcal{A}}}
\newcommand{\bc}{\ensuremath{\mathcal{B}}}

\newcommand{\vb}{\ensuremath{\mathbf{v}}}
\newcommand{\eb}{\ensuremath{\mathbf{e}}}
\newcommand{\fb}{\ensuremath{\mathbf{f}}}
\newcommand{\ub}{\ensuremath{\mathbf{u}}}
\newcommand{\rb}{\ensuremath{\mathbf{r}}}
\newcommand{\cb}{\ensuremath{\mathbf{c}}}

\def \< {\langle}
\def \> {\rangle}

\newcommand{\ent}[1]{{\left[{#1}\right]}}
\newcommand{\abs}[1]{{\left|{#1}\right|}}
\newcommand{\scal}[1]{{\left\langle{#1}\right\rangle}}

\begin{document}

\title[A Fourier analytic approach to MUBs]{A Fourier analytic approach to the problem of mutually unbiased bases}

\author[M. Matolcsi]{M\'at\'e Matolcsi}
\address{M. M.: Alfr\'ed R\'enyi Institute of Mathematics,
Hungarian Academy of Sciences POB 127 H-1364 Budapest, Hungary
Tel: (+361) 483-8302, Fax: (+361) 483-8333}
\email{matomate@renyi.hu}

\thanks{M. Matolcsi was supported by the ERC-AdG 228005, and OTKA Grants No. K77748, K81658.}


\begin{abstract}
We give an entirely new approach to the problem of mutually unbiased bases (MUBs), based on a Fourier analytic technique in additive combinatorics. The method provides a short and elegant generalization of the fact that there are at most $d+1$ MUBs in $\Co^d$. It may also yield a proof that no complete system of MUBs exists in some composite dimensions -- a long standing open problem.
\end{abstract}

\maketitle

\bigskip



\section{Introduction}

In this paper we introduce a novel approach to the problem of mutually unbiased bases in $\Co^d$. Surprisingly enough, the required Fourier analytic technique is borrowed from additive combinatorics -- a seemingly unrelated branch of mathematics.

\medskip

The paper is organized as follows. In the Introduction we recall some basic notions and results concerning mutually unbiased bases (MUBs). In Section \ref{sec2} we describe how the problem of MUBs fits into a general scheme in additive combinatorics -- a scheme we will call {\it Delsarte's method.} We then apply this method to prove Theorem \ref{thmgeneral}, an elegant generalization of the fact that there are at most $d+1$ MUBs in $\Co^d$. Finally, in Section \ref{sec3} we indicate the limitations of the method by introducing the notion of pseudo-MUBs, and discuss the possible existence of such in the case $d=6$.

\medskip

Recall that given an orthonormal basis $\ac=\{\eb_1,\ldots,\eb_d\}$ in $\Co^d$, a unit vector $\vb$ is called {\it unbiased} to $\ac$ if
$\abs{\scal{\vb,\eb_k}}=\dst\frac{1}{\sqrt{d}}$ for all $1\le k\le d$.
Two orthonormal bases in $\Co^d$,
$\ac=\{\eb_1,\ldots,\eb_d\}$ and $\bc=\{\fb_1,\ldots,\fb_d\}$ are
called \emph{unbiased} if for every $1\leq j,k\leq d$,
$\abs{\scal{\eb_j,\fb_k}}=\dst\frac{1}{\sqrt{d}}$. A collection
$\bc_0,\ldots\bc_m$ of orthonormal bases is said to be
\emph{(pairwise) mutually unbiased} if every two of them are
unbiased. What is the maximal number of pairwise mutually unbiased bases (MUBs) in $\Co^d$? This question originates from quantum information theory and has been investigated thoroughly over the past decades (see \cite{durt} for a recent comprehensive survey on MUBs). The following result is well-known (see e.g. \cite{BBRV,BBELTZ,WF}):

\begin{theorem}\label{thm1}
The number of mutually unbiased bases in $\Co^d$ cannot exceed $d+1$.
\end{theorem}

We will generalize this fact in Theorem \ref{thmgeneral} below. The other important well-known result concerns prime-power dimensions (see e.g.
\cite{BBRV,Com0,Com1,Com2,Iv,KR,WF}).

\begin{theorem}\label{thm2}
A collection of $d+1$ mutually unbiased bases (called a {\it complete set} of MUBs) can be constructed
if the dimension $d$ is a prime or a prime-power.
\end{theorem}

However, if the dimension $d=p_1^{\alpha_1}\dots p_k^{\alpha_k}$ is composite then very little
is known except for the fact that there are at least $p_j^{\alpha_j}+1$ mutually
unbiased bases in $\C^d$ where $p_j^{\alpha_j}$ is the smallest of the prime-power divisors. In some specific square dimensions there is also a construction based on orthogonal Latin squares which yields more MUBs than $p_j^{\alpha_j}+1$ (see \cite{262}).
The following basic problem, however, remains open for all non-primepower dimensions:

\begin{problem}\label{MUB6problem}\
Does a complete set of $d+1$ mutually unbiased bases exist in $\Co^d$ if $d$ is not a prime-power?
\end{problem}

The answer is not known even for $d=6$, despite considerable efforts over the past few years
(\cite{BBELTZ,config,ujbrit,arxiv,Msz,Skinner,boykin}). The case $d=6$ is particularly tempting because it seems to be the simplest to handle with algebraic and numerical methods.  As of now, some {\it infinite families} of MUB-triplets
in $\C^6$ have been constructed (\cite{Za,arxiv}), but numerical evidence suggests that there exist no MUB-quartets \cite{config,ujbrit,numerical,Za}.

\medskip

It will also be important for us to recall that mutually unbiased bases
are naturally related to \emph{complex Hadamard matrices}.
Indeed, if the bases $\bc_0,\ldots,\bc_m$ are mutually
unbiased we may identify each
$\bc_l=\{\eb_1^{(l)},\ldots,\eb_d^{(l)}\}$ with the \emph{unitary}
matrix
$$
[H_l]_{j,k}=\ent{\scal{\eb_j^{(0)},\eb_k^{(l)}}_{1\leq k,j\leq
d}},
$$
{\it i.e.} the $k$-th column of $H_l$ consists of the
coordinates of the $k$-th vector of $\bc_l$ in the basis $\bc_0$.
(Throughout the paper the scalar product $\scal{.,.}$ of $\C^d$ is
conjugate-linear in the first variable and linear in the second.) With this convention,
$H_0=I$ the identity matrix and all other matrices are unitary
and have entries of modulus $1/\sqrt{d}$. Therefore, the matrices
$H'_l=\sqrt{d}H_l$ have all entries of modulus 1 and complex orthogonal
rows (and columns). Such matrices are called \emph{complex
Hadamard matrices}. It is thus clear that the existence of a family of
mutually unbiased bases $\bc_0,\ldots,\bc_m$ is equivalent to
the existence of a family of complex Hadamard matrices
$H'_1,\ldots, H'_m$ such that for all
 $1\leq j\not=k\leq
m$, $\frac{1}{\sqrt{d}}H^{'*}_jH'_k$ is again a complex Hadamard matrix. In such
a case we will say that these complex Hadamard matrices are {\it
mutually unbiased}.

\medskip

A complete classification of MUBs up to dimension 5 (see \cite{BWB}) is based on the classification of complex Hadamard matrices (see \cite{haagerup}). However, the classification of complex Hadamard matrices in dimension 6 is still out of reach despite recent
efforts \cite{BN,Msz,Skinner,star,karlsson}.

\medskip

In this paper we will use the above connection of MUBs to complex Hadamard matrices to apply a Fourier analytic approach, borrowed from additive combinatorics.

\section{Mutually unbiased bases, difference sets and Delsarte's method}\label{sec2}

In this section we describe a general scheme in additive combinatorics, and show how the problem of mutually unbiased bases fit into this scheme.

\medskip

Let $G$ be a compact Abelian group, and let a symmetric subset $A=-A\subset G$, $0\in A$ be given. We will call $A$ the 'forbidden' set. We would like  to determine the maximal cardinality of a set $B=\{b_1, \dots b_m\}\subset G$ such that all differences $b_j-b_k\in A^c\cup \{0\}$ (in other words,  all differences avoid the forbidden set $A$). Some well-known examples of this general scheme are present in coding theory (\cite{delsarte}), sphere-packings (\cite{cohnelkies}), and sets avoiding square differences in number theory (\cite{ruzsa}).

\medskip

We now describe a general method to tackle such problems. To the best of my knowledge it was first introduced by Delsarte (in a less general terminology) in connection with binary codes with prescribed Hamming distance. The method is also 'folklore' in the additive combinatorics community and I was introduced to it by Imre Z. Ruzsa (\cite{ruzsapersonal}).

\medskip

We are looking for a 'witness' function $h: G\to \R$ with the following properties.

\medskip

$\bullet$ $h$ is an even function, $h(x)=h(-x)$, such that the Fourier inversion formula holds for  $h$ (in particular, $h$ can be any finite linear combination of characters on $G$).

\medskip

$\bullet$ $h(x)\le 0$ for all $x\in A^c$

\medskip

$\bullet$ $\hat{h}(\gamma)\ge 0$ for all $\gamma\in \hat{G}$

\medskip

$\bullet$ $\hat{h}(0)=1$.

\begin{lemma}\label{delsartelemma} (Delsarte's method)\\
Given a function $h: G\to \R$ with the properties above, we can conclude that for any $B=\{b_1, \dots b_m\}\subset G$ such that $b_j-b_k\in A^c\cup \{0\}$ the cardinality of $B$ is bounded by $|B|\le h(0)$.
\end{lemma}

\begin{proof}
For any $\gamma\in \hat{G}$ define $\hat{B}(\gamma)=\sum_{j=1}^m \gamma (b_j)$. Now, evaluate
\begin{equation}\label{eq1}
S=\sum_{\gamma\in \hat{G}} |\hat{B}(\gamma)|^2 \hat{h}(\gamma).
\end{equation}

All terms are nonnegative, and the term corresponding to $\gamma=0$ (the trivial character, i.e. $\gamma(x)=1$ for all $x\in G$) gives $|\hat{B}(0)|^2\hat{h}(0)=|B|^2$. Therefore

\begin{equation}\label{eq2}
S\ge |B|^2.
\end{equation}

On the other hand, $|\hat{B}(\gamma)|^2= \sum_{j,k}\gamma(b_j-b_k)$, and therefore $S=\sum_{\gamma,j,k}\gamma(b_j-b_k)\hat{h}(\gamma)$. Summing up for fixed $j,k$ we get \\
$\sum_{\gamma}\gamma(b_j-b_k)\hat{h}(\gamma)=h(b_j-b_k)$ (the Fourier inversion formula), and therefore $S=\sum_{j,k} h(b_j-b_k)$. Notice that $j=k$ happens $|B|$-many times, and all the other terms (when $j\ne k$) are non-positive because $b_j-b_k\in A^c$, and $h$ is required to be non-positive there. Therefore
\begin{equation}\label{eq3}
S\le h(0)|B|.
\end{equation}
Comparing the two estimates \eqref{eq2}, \eqref{eq3} we obtain $|B|\le h(0)$.
\end{proof}

\medskip

How do mutually unbiased bases fit into this scheme? The answer is that they almost perfectly do, except for the fact that the underlying group is not Abelian. Indeed, let $G=U_{d\times d}$ the group of $d\times d$ unitary matrices, and let $H\subset U_{d\times d}$ denote the set of complex Hadamard matrices (rescaled by the factor $1/\sqrt{d})$) in $U_{d\times d}$. Let the 'forbidden' set $A$ be the complement of $H$. Of course, the group operations $+$ and $-$ in the Delsarte scheme are now replaced by matrix multiplication and inverse. Also, the role of zero element is taken by the identity matrix. Then, the maximal number of mutually unbiased bases in $\Co^d$ is exactly the maximal cardinality of a set $\{U_0,\ U_1, \dots U_m\} \subset G$ such that all 'differences' $U_j^\ast U_k$ ($0\le j,k\le m$) lie in the prescribed subset $A^c\cup \{I\}$.

\medskip

Unfortunately, we do not know how to generalize Delsarte's method to the case of non-commutative groups, in particular to $G=U_{d\times d}$. Nevertheless, we can still use Delsarte's scheme if we rephrase the problem appropriately, as follows.

\medskip

Assume that a family $H_1, \dots H_m$ of $m$ mutually unbiased complex Hadamard matrices exists.
Then all entries of all matrices are of modulus $1$, and the columns (and
thus the rows) within each matrix are complex orthogonal, and we have the unbiasedness
condition: for any two columns $\ub,\vb$ coming from different matrices
we have $|\langle \ub, \vb \rangle |=\sqrt{d}$. (Recall that we have re-normalized the matrices by a factor of $\sqrt{d}$.)

\medskip

After multiplying rows and columns by appropriate scalars if
necessary, we can assume that all coordinates of the first row and
column of $H_1$ are 1's, and all coordinates of the first row of
all other matrices are 1's (i.e. we assume that all appearing
columns have first coordinate 1, and the
first column in $H_1$ consists of 1's. This is standard and trivial normalization.) All the other
coordinates in the matrices are complex numbers of modulus 1, i.e.
they are of the form $e^{2\pi i \rho}$ with $\rho\in [-1/2,1/2)$. Therefore, we can associate to each column vector
$(1, e^{2\pi i \rho_1}, \dots ,e^{2\pi i \rho_{d-1}})$ the vector $(0, \rho_1, \dots, \rho_{d-1})\in \TT^d$, the real $d$-dimensional torus, $\TT^d=[-1/2,1/2)^d$. Also, note that the first coordinate always automatically becomes 0, because each column starts with coordinate 1. Therefore we make the more useful association that a column ${\cb}=(1, e^{2\pi i \rho_1}, \dots ,e^{2\pi i \rho_{d-1}})$ is represented by ${\ub}=(\rho_1, \dots \rho_{d-1})\in \TT^{d-1}$, the $d-1$-dimensional torus. There are altogether $md$ column vectors in the Hadamard matrices $H_1, \dots H_m$, and we will denote the associated vectors in $\TT^{d-1}$ by $\ub_1, \dots \ub_{md}$ (we will see that in this approach it is not really relevant to indicate which vector comes from which basis. But let us agree for convenience that $\ub_1=(0, \dots, 0)$, corresponding to the first column of $H_1$.)

\medskip

Two columns ${\cb_1}=(1, e^{2\pi i \rho_1}, \dots ,e^{2\pi i \rho_{d-1}})$ and ${\cb_2}=(1, e^{2\pi i \mu_1}, \dots ,e^{2\pi i \mu_{d-1}})$ are orthogonal if and only if $1+\sum_{j=1}^{d-1} e^{2\pi i (-\rho_j+\mu_j)}=0$, and they are unbiased if and only if $|1+\sum_{j=1}^{d-1} e^{2\pi i (-\rho_j+\mu_j)}|=\sqrt{d}$. Therefore it is natural to introduce the following definitions.

\begin{definition}
Let $ORT_d$ denote the set of vectors $(\alpha_1, \dots \alpha_{d-1})\in \TT^{d-1}$, in the $d-1$-dimensional torus, such that $1+\sum_{j=1}^{d-1} e^{2\pi i \alpha_j}=0$. Also, let $UB_d$ denote the set of vectors $(\alpha_1, \dots \alpha_{d-1})\in \TT^{d-1}$, such that $|1+\sum_{j=1}^{d-1} e^{2\pi i \alpha_j}|=\sqrt{d}$.
Let us also define the 'forbidden' set $A_d=(ORT_d \cup UB_d)^c$.
\end{definition}

\medskip

We conclude that the vectors $\ub_1, \dots \ub_{md}$ satisfy that the difference of any two of them (the difference being taken $mod \ 1$ in each coordinate, i.e. we take the difference in the group $\TT^{d-1}$) lies in $A_d^c\cup \{0\}$. Therefore, we have arrived exactly to the scheme of Lemma \ref{delsartelemma}.

\medskip

As a preliminary remark we note that the dual group of $G=\TT^{d-1}$ is $\hat{G}=\ZZ^{d-1}$. And the action of a character $\gamma\in \ZZ^{d-1}$ on a point $x\in\TT^{d-1}$ is given as $\gamma(x)=e^{2\pi i \langle \gamma, x\rangle}$. In particular, $\gamma=0$ is the trivial character (constant 1). The Fourier transform of a function $f: G\to \Co$ is a function $\hat{f}: \hat{G}\to \Co$ given as $\hat{f}(\gamma)= \int_{x\in G}f(x)\gamma(x)dx$.

\medskip

Let us see whether we can find a good 'witness' function in this situation. At first sight things do not look promising because we have no understanding of the geometry of the sets $ORT_d$ and $UB_d$ inside the torus $\TT^{d-1}$. However, it turns out such geometric understanding is not necessarily required and we easily prove the following generalization of Theorem \ref{thm1}.

\begin{theorem}\label{thmgeneral}
Let $\ac$ be an orthonormal basis in $\Co^d$, and let $B=\{\cb_1, \dots \cb_r\}$ consist of unit vectors which are all unbiased to $\ac$. Assume that for all $1\le j\ne k \le r$ the vectors $\cb_j$ and $\cb_k$ are either orthogonal or unbiased to each other, i.e. either $\langle \cb_j, \cb_k \rangle =0$ or
$|\langle \cb_j, \cb_k \rangle |=1/\sqrt{d}$. Then $r\le d^2$.
\end{theorem}

\begin{proof}
As we saw in the discussion above, the vectors $\ub_1, \dots \ub_{r}\in \TT^{d-1}$ (associated to $\sqrt{d}\cb_1, \dots \sqrt{d}\cb_r$) satisfy $\ub_j-\ub_k\in A_d^c\cup \{0\}$ for all $1\le j,k\le r$. Therefore Lemma \ref{delsartelemma} can be applied.

\medskip

Define the 'witness' function $h:\TT^{d-1} \to \R$ as follows:
\begin{equation}\label{h}
h(x_1, \dots x_{d-1})=\frac{1}{(d-1)d}\left|1+\sum_{j=1}^{d-1} e^{2\pi i x_j}\right|^2\left (|1+\sum_{j=1}^{d-1} e^{2\pi i x_j}|^2-d \right ).
\end{equation}

It is trivial to check that $h$ satisfies all requirements. Indeed, $h$ is an even function which vanishes on $ORT_d \cup UB_d$.  The Fourier coefficients of $h$ are simply the coefficients of the exponential terms after expanding the brackets, and these are clearly nonnegative.  Also $\hat{h}(0)=1$ because $\hat{h}(0)$ is the integral of $h$, which is just the constant term.  Also, $h(0, \dots 0)=d^2$, so that we conclude from Lemma \ref{delsartelemma} that $|B|\le d^2$.
\end{proof}

\medskip

\begin{remark}\rm
As shown by Theorem \ref{thm2} the result of Theorem \ref{thmgeneral} is sharp if $d$ is a prime-power. However, if $d$ is not a prime-power, then it could be possible to find a better witness function than above. The function $h$ above uses simply the {\it definition} of the sets $ORT_d$ and $UB_d$. In principle, it could be possible to find some structural properties of these sets in dimension 6 (or any other composite dimension), in order to construct a better witness function and get a sharper bound on $r$. Any upper bound $r<d^2$ would mean that a complete set of MUBs does not exist in dimension $d$. We have not been able to make such improvement for any $d$ so far. \hfill $\square$
\end{remark}

Another observation is that if $r=d^2$ in Theorem \ref{thmgeneral} then both estimates \eqref{eq2}, \eqref{eq3} must hold with equality.
On the one hand, it is trivial that \eqref{eq3} automatically becomes an equality for the $h$ above (because $h$ is zero on $ORT_d$ and $UB_d$). On the other hand, inequality \eqref{eq2} becomes an equality {\it if only if} $|\hat{B}(\gamma)|^2 \hat{h}(\gamma)=0$ for all $\gamma\ne 0$. These are non-trivial conditions and we obtain the following corollary, which is a generalization of Theorem 8 in \cite{belovs}.

\begin{corollary}
Let $\ac$ be an orthonormal basis in $\Co^d$, and let $B=\{\cb_1, \dots \cb_{d^2}\}$ consist of unit vectors which are all unbiased to $\ac$. Assume that for all $1\le j\ne k \le d^2$ the vectors $\cb_j$ and $\cb_k$ are either orthogonal or unbiased to each other. Write $B$ as a $d\times d^2$ matrix, the columns of which are the vectors $\cb_j$, $j=1, \dots d^2$. Let $\rb_1,\dots \rb_d$ denote the rows of the matrix $B$, and let $\rb_{j/k}=\rb_j/\rb_k$ denote the coordinate-wise quotient of the rows. Then the vectors $\rb_{j/k}$ ($1\le j\ne k\le d$) are orthogonal to each other in $\Co^{d^2}$, and they are all orthogonal to the vector $(1,1,\dots 1)\in \Co^{d^2}$.
\end{corollary}

\begin{proof}
This is a direct consequence of the proof of Lemma \ref{delsartelemma}. Indeed, for \eqref{eq2} to be an equality  $\hat{B}(\gamma)$ must be zero whenever $\hat{h}(\gamma)\ne 0$ and $\gamma\ne 0$. Looking at the definition of $h$, $\hat{h}(\gamma)\ne 0$ happens exactly when $\gamma (x_1, \dots x_{d-1})=e^{2\pi i ((x_j-x_k)+(x_q-x_s))}$, for any quadruple $0\le j,k,q,s\le d-1$, where we use the convention that $x_0=0$. Using the definition of $\hat{B}(\gamma)$ we obtain that $\hat{B}(\gamma)=0$ means exactly that $\rb_{j/k}$ and $\rb_{s/q}$ are orthogonal to each other.

\end{proof}

\begin{remark}\rm
There are altogether $d(d-1)$ vectors of the form $\rb_{j/k}$, and with the addition of $(1,1,\dots 1)$ we get a system of $d(d-1)+1$ orthogonal vectors in $\Co^{d^2}$. It is not at all obvious whether in each dimension $d^2$ there exists a set of vectors $R=\{\rb_1, \dots \rb_d\}$, with all coordinates having absolute value 1, such that they satisfy these orthogonality constraints. Let us remark that if modulo $d^2$ a $d$-element Sidon-set exists  (i.e. a set where the non-zero differences are all distinct), then the corresponding rows of the Fourier matrix $F_{d^2}$ form an appropriate system $R$. For example, for $d=6$ the set $\{0, 1, 3, 8, 23, 27\}$ is a Sidon-set modulo 36, so the rows $R=\{\fb_0, \fb_1, \fb_3, \fb_8, \fb_{23}, \fb_{27}\}$ of the Fourier matrix $F_{36}$ satisfy all orthogonality constraints (but $R$ is not a concatenation of orthonormal bases, so it does not form a complete set of MUBs). One could try to find a composite $d$ such that a $d$-element Sidon set modulo $d^2$ does not exist, and then show that an appropriate set of vectors $R$ also cannot exist. This would prove that a complete set of MUBs does not exist in dimension $d$.  \hfill $\square$
\end{remark}

\section{Linear duality and pseudo-MUBs}\label{sec3}

We can view the set $B$ in Lemma \ref{delsartelemma} as a 0-1-valued function on $G$. Also, observe that $B$ does not directly enter the proof, but instead the function $|\hat{B}(\gamma)|^2=\widehat{B-B}(\gamma)$ is essential. For any $y\in \TT^{d-1}$ let $f(y)$ denote the number of ways of writing $y$ as a difference of two elements of $B$. Then, for any $\gamma\in \hat{G}=\ZZ^{d-1}$ we have $|\hat{B}(\gamma)|^2= \sum_{j,k}\gamma(b_j-b_k)=\sum_{y}f(y)e^{2\pi i\langle \gamma,y\rangle}$. Therefore, $f$ has the following essential properties.

\medskip

$\bullet$ the finite exponential sum $\sum_{y}f(y) e^{2\pi i\langle \gamma,y\rangle}$ is nonnegative for all $\gamma\in\ZZ^{d-1}$, and the exponents $y\in A^c\cup\{0\}=ORT_d\cup UB_d\cup \{0\}$. We can view it as the Fourier transform of the function $f : \TT^{d-1}\to \R_+$.

\medskip

$\bullet$ the coefficients $f(y)$ are nonnegative integers.

\medskip

$\bullet$ the sum of the coefficients $\sum_y f(y)= |B|^2$.

\medskip

$\bullet$ $f(0)=|B|$.

\medskip

Given any such function $f$ we can repeat the proof of Lemma \ref{delsartelemma} with the 'witness' function $h$ defined in \eqref{h}, and conclude that
\begin{equation}
\frac{\sum_y f(y)}{f(0)}\le \frac{h(0)}{\hat{h}(0)}=d^2.
\end{equation}

This motivates the following definition:

\begin{definition}We will call a function $f: \TT^{d-1}\to \R_+$ a complete-pseudo-MUB-system in dimension $d$ (or pseudo-MUB-$d$ in short) if it satisfies the following conditions:

\medskip

$\bullet$ $f$ is nonnegative, the support of $f$ is finite and is contained in $ORT_d\cup UB_d\cup \{0\}$

\medskip

$\bullet$ the finite exponential sum $\hat{f}(\gamma)=\sum_{y}f(y) e^{2\pi i\langle \gamma,y\rangle}$ is nonnegative for all $\gamma\in\ZZ^{d-1}$.

\medskip

$\bullet$ the sum of the coefficients $\sum_y f(y)= d^4$.

\medskip

$\bullet$ $f(0)=d^2$.
\end{definition}

Notice that $f$ is not required to be integer valued. As discussed above, a complete set of $d+1$ MUBs always gives rise to a complete-pseudo-MUB-system. Indeed, let $B\subset \TT^{d-1}$ denote the $d^2$ columns of the corresponding $d$ mutually unbiased Hadamards, and let $f(y)=(B-B)(y)$, meaning the number of ways $y$ can be written as a difference $b_j-b_k$. Then $f$ is a complete-pseudo-MUB-system. The converse is not necessarily true: a complete-pseudo-MUB-system does not directly imply the existence of a complete set of MUBs.

\medskip

In any dimension $d$, if we find a pseudo-MUB-$d$ function $f$ then it could serve as a 'dual-witness' testifying that our function $h$ in equation \eqref{h} is best possible, and it would mean that the Delsarte method {\it alone} cannot prove the non-existence of $d+1$ MUBs in dimension $d$. We emphasize that it would {\it not} mean that a complete system of $d+1$ MUBs exists. It would only mean that a complete-pseudo-MUB-system exists.

\medskip

We remark that there is a linear duality here. {\it Either} a pseudo-MUB-$d$ exists, {\it or} a better witness function $h:\TT^{d-1}\to \R$ exists, proving that $r<d^2$ in Theorem \ref{thmgeneral}. There is no third option! In the latter case we could conclude that no complete system of $d+1$ MUBs exists in dimension $d$. In the former case we would have an interesting pseudo-MUB-$d$ in our hand, which could possibly lead later to the discovery of a proper complete set of MUBs.

\medskip

Let us examine the situation in dimension $d=6$.

\medskip

One natural idea is to fix some $m$, and look for a pseudo-MUB-6 function such that its support contains vectors only whose coordinates are $m$th roots of unity. The reason is that all known complete sets of MUBs consist of such vectors. It is also convenient because such vectors belonging to $ORT_6 $ and $UB_6$ can easily be listed by a computer code. Furthermore, the restriction $\hat{f}(\gamma)\ge 0$ needs to be checked only as $\gamma$ ranges over the cube $[0,m-1]^5$, due to periodicity. Finally, we fix $f(0)=1$ (which is a somewhat more convenient normalization than $f(0)=d^2$ in the definition), and maximize $M={\sum_y f(y)}$ by linear programming (a pseudo-MUB-6 would have the value 36 here).   We have tried this and the results are the following:

\medskip

$\bullet$ $m=12$, $M=17.5$

\medskip

$\bullet$ $m=8$, $M=21.6$

\medskip

$\bullet$ $m=16$, $M=21.6$

\medskip

Larger values of $m$ are out of our computational power. As one can see, these results are inconclusive. We could not find a pseudo-MUB-6 but we could not find a better 'witness' function $h(x)$ either in dimension 6. By linear duality one of them must exit, and it would be interesting to see which one.

\end{document}